\documentclass[12pt]{amsart}

\usepackage{graphicx}
\usepackage{algorithm}
\usepackage{algorithmic}
\usepackage{float}
\usepackage[caption = false]{subfig}

\newtheorem{theorem}{Theorem}
\newtheorem{thmalpha}{Theorem}
\newtheorem{lemma}{Lemma}
\newtheorem{proposition}{Proposition}

\newtheorem{definition}{Definition}

\theoremstyle{remark}

\newtheorem{example}{Example}

\newcommand{\ltwo}{L^2({\mathbb R})}

\begin{document}


\title{Conjugate Phase Retrieval in Paley-Wiener Space}
\author{Chun-Kit Lai}
\address{Department of Mathematics, San Francisco State University, San Francisco, CA 94132}
\email{cklai@sfsu.edu}
\author{Friedrich Littmann}
\address{Department of Mathematics, North Dakota State University, Fargo, ND 58105}
\email{friedrich.littmann@ndsu.edu}
\author{Eric S. Weber}
\address{Department of Mathematics, Iowa State University, 396 Carver Hall, Ames, IA 50011}
\email{esweber@iastate.edu}
\date{\today}
\begin{abstract}
We consider the problem of conjugate phase retrieval in Paley-Wiener space $PW_{\pi}$.  The goal of conjugate phase retrieval is to recover a signal $f$ from the magnitudes of linear measurements up to unknown phase factor and unknown conjugate, meaning $f(t)$ and $\overline{f(t)}$ are not necessarily distinguishable from the available data.  We show that conjugate phase retrieval can be accomplished in $PW_{\pi}$ by sampling only on the real line by using structured convolutions. We also show that conjugate phase retrieval can be accomplished in $PW_{\pi}$ by sampling both $f$ and $f^{\prime}$ only on the real line.  Moreover, we demonstrate experimentally that the Gerchberg-Saxton method of alternating projections can accomplish the reconstruction from vectors that do conjugate phase retrieval in finite dimensional spaces.  Finally, we show that generically, conjugate phase retrieval can be accomplished by sampling at three times the Nyquist rate, whereas phase retrieval requires sampling at four times the Nyquist rate.
\end{abstract}
\maketitle

\section{Introduction}

The phase retrieval problem can be stated as follows:  can a signal $f$ be reconstructed from the magnitudes of linear measurements of $f$?  Naturally, $f$ and $\alpha f$ cannot be distinguished by the magnitudes of linear measurements, where $\alpha$ is any scalar of magnitude $1$.  In general, one wishes to design a sampling scheme so that the magnitudes of linear measurements can distinguish all signals up to the ambiguity of this uniform phase factor.  We consider in the present paper a weaker formulation of the problem: can a signal $f$ be reconstructed from the magnitudes of linear measurements, up to the ambiguity of $\alpha f$ and $\alpha \overline{f}$?  We refer to this as the conjugate phase retrieval problem.

Let us make precise our problem formulation here.  The Paley-Wiener space $PW_{\gamma}$ consists of all $f \in \ltwo$ such that $\hat{f}(\xi) = 0$ for a.e. $\xi \in \mathbb{R} \setminus [-\gamma, \gamma]$.  Here, $\gamma$ is any positive number.  Any $f \in PW_{\gamma}$ has an extension to an entire function on the complex plane.  Moreover, if $f \in PW_{\gamma}$, then the entire function $f^{\sharp}$ defined by $f^{\sharp}(z) = \overline{f(\bar{z})}$ is in $PW_{\gamma}$ as well.  We define an equivalence relation on $PW_{\gamma}$ as follows:  for $f,g \in PW_{\gamma}$ 
\begin{equation} \label{Eq:equiv}
\qquad f \sim g \text{ if } f = \lambda g, \text{ or } f = \lambda g^{\sharp} \text{ for some $| \lambda | = 1$. } 
\end{equation}
Our goals are as follows:
\begin{enumerate}
\item[(a)] design a sequence of linear functionals (measurements) $\phi_{n} : PW_{\gamma} \to \mathbb{C}$ such that the mapping from $PW_{\gamma}/\sim$ to $\ell^2(\mathbb{Z})$ given by 
\begin{equation*}
  f \mapsto ( | \phi_{n} (f) | )_{n} 
\end{equation*}
is one-to-one,
\item[(b)]  reconstruct $[f]$ from $( | \phi_{n}(f) | )_n$, where $[f]$ denotes the equivalence class in $PW_{\gamma}/ \sim$ of $f \in PW_{\gamma}$.
\end{enumerate}

The phase retrieval problem originates in optics \cite{gerchberg1972practical,Fienup78,Ros84a,LS84a,KliSac1992phaseless}.  Modern phase retrieval is often considered in the case of frames \cite{BCE06a,BBCE09a,BCMN14a}.  Conjugate phase retrieval for frames was introduced in \cite{EL2017} (see also \cite{chen2019phase,yang2019generalized}).  Phase retrieval in the context of wavelets and other systems appear in \cite{MW15a,CCD16a,ADGT17a}.  Phase retrieval in the Paley-Wiener space in particular is discussed in \cite{Tha11a,PYB14a}.  In \cite{Tha11a} considers the case of real phase retrieval in $PW_{\pi}$, meaning only real-valued signals $f$ are sampled.  The main result is that if one samples $f$ at twice the Nyquist frequency, then $\pm f$ can be recovered from $(| f(t_{n}) |)_{n}$.  We note here that the reconstruction of $\pm f$ given in \cite{Tha11a} involves reconstruction off of the real axis.  Similarly, \cite{PYB14a} considers the case of (complex) phase retrieval in $PW_{\pi}$ by designing a sampling scheme that occurs off of the real axis.  In particular, the sampling scheme as presented in \cite{PYB14a} takes the form
\begin{equation} \label{Eq:struct-conv}
\phi_{n} (f) = \sum_{j} c_{j,n} f(z_{n} + b_{j,n})
\end{equation}
for complex scalars $c_{j,n}, z_{n}, b_{j,n}$.  Sampling schemes such as this are referred to as \emph{structured modulations} in \cite{PYB14a} because the authors there consider the reconstruction in the Fourier domain, where the shifts become modulations. 

\section{Conjugate Phase Retrieval}
We will design sampling schemes for the conjugate phase retrieval problem in $PW_{\pi}$ (our statements can be modified appropriately for $PW_{\gamma}$).  In Subsection \ref{ssec:conv}, our sampling scheme will take the form of structured convolutions.  However, we will demonstrate that by solving the conjugate phase retrieval problem (which is weaker than the phase retrieval problem), we will be able to both sample and perform the reconstruction on the real axis.  In Subsection \ref{ssec:deriv}, we will show that the conjugate phase retrieval problem can be solved by sampling both $f$ and $f^{\prime}$ (on the real axis as well) rather than with structured convolutions.

\subsection{Qualitative Results}

While our main focus of the paper is to demonstrate reconstruction algorithms, we first prove qualitative results concerning conjugate phase retrieval on the Paley-Wiener space.  In particular, for the choice $\phi_n(f) = f(t_n+b) - f(t_n)$ we determine in Theorem \ref{Th:obvious} when the corresponding mapping on $PW_{\gamma}/\sim$ is injective and has a continuous inverse.   Our proofs are based on several elementary and known results.  The first elementary result concerns the square of a signal $f \in PW_{\gamma}$:
\begin{lemma} \label{L:pw-closed}
If $f \in PW_{\gamma}$, then:
\begin{enumerate}
\item $f^{\prime} \in PW_{\gamma}$;
\item $f f^{\sharp} \in PW_{2 \gamma}$;
\item $f^{\prime} (f^{\prime})^{\sharp} \in PW_{2 \gamma}$.
\end{enumerate}
\end{lemma}

The known result we need is the following \cite[Theorem 3]{McD04a}:
\begin{thmalpha} \label{Th:McD}
Suppose $f,g \in PW_{\gamma}$.
\begin{enumerate}
\item If $0< b < \gamma/{\pi}$, and for all $x \in \mathbb{R}$, $|f(x)| = |g(x)|$ and $|f(x + b) - f(x)| = |g(x + b) - g(x)|$, then $f \sim g$.
\item If for all $x \in \mathbb{R}$, $|f(x)| = |g(x)|$ and $|f^{\prime}(x)| = | g^{\prime}(x)|$, then $f \sim g$.
\end{enumerate}
\end{thmalpha}
In Theorem \ref{Th:McD}, $f \sim g$ is the equivalence relation given in Equation (\ref{Eq:equiv}).

Recall that a sequence $\{ t_{n} \}_{n} \subset \mathbb{R}$ is a \emph{set of sampling} for $PW_{\gamma}$ provided that there exist constants $0< A,B$ such that
\[ A \| f \|^2 \leq \sum_{n} | f(t_{n})|^2 \leq B \| f \|^2 \]
holds for all $f \in PW_{\gamma}$.  For a set of sampling, there exists a dual sequence $\{ g_{n} \}_{n \in \mathbb{Z}} \subset PW_{\gamma}$ such that 
\begin{equation} \label{Eq:interpolate}
 f(t) = \sum_{n \in \mathbb{Z}} f(t_{n}) g_{n}(t)
\end{equation}
with convergence holding both pointwise and in $PW_{\gamma}$-norm.  See \cite{DS52a,HS99a,BF01a,Web04b} for more details.

We immediately obtain the following theorem:
\begin{theorem} \label{Th:obvious}
Suppose $\{ t_{n} \} \subset \mathbb{R}$ is a set of sampling for $PW_{2 \gamma}$.  Then the mapping $\mathcal{A} : PW_{\gamma} / \sim \to \ell^{2}(\mathbb{Z}) \oplus \ell^{2}(\mathbb{Z})$ defined by
\begin{equation*} 
\mathcal{A}(f) =  ( | f(t_{n})|, | f(t_{n} + b) - f(t_n)|)_{n\in\mathbb{Z}}
\end{equation*}
is one-to-one whenever $0 < b < 2 \gamma$, and the mapping $\widetilde{\mathcal{A}} : PW_{\gamma}/ \sim \to \ell^{2}(\mathbb{Z}) \oplus \ell^{2}(\mathbb{Z})$ defined by
\begin{equation*}
\widetilde{\mathcal{A}}(f) = ( | f(t_{n})|, | f^{\prime}(t_{n})|)_{n\in\mathbb{Z}}
\end{equation*}
is one-to-one.
\end{theorem}

The proof follows from the fact that $f f^{\sharp}(t)$ and $(f(t + b) - f(t))(f(t+b) - f(t))^{\sharp}$ can be reconstructed from the sequence of samples $(| f(t_{n})|^2)_{n}$ $( | f(t_{n} + b) - f(t_{n})|^2)_{n}$, respectively, which we note can be done in a  stable way from the hypotheses.  While the theorem guarantees the invertibility of $\mathcal{A}$, there is no obvious algorithm for actually reconstructing $[f]$ from $\mathcal{A}(f)$ (or $\widetilde{\mathcal{A}}(f)$).  The only potential reconstruction given by the proof of Theorem \ref{Th:McD} utilizes Hadamard factorizations of entire function of finite type, which requires knowledge of the zeros of the function.  We are unaware of numerical methods to find the zeros of the unknown function $f$ from $\mathcal{A}(f)$.  We will demonstrate a numerical reconstruction algorithm for $\mathcal{A}$ (Theorem \ref{Th:main1} and Algorithm \ref{RA:conv}) at the cost of needing to sample more than just the two functions $f(t)$ and $f(t + b) - f(t)$, and we will demonstrate an alternative theoretical reconstruction algorithm for $\widetilde{\mathcal{A}}$ (Algorithm \ref{RA:derivative}).


\begin{theorem}
The range $\mathcal{R}(\mathcal{A})$ is closed.  The inverse $\mathcal{A}^{-1}$ is continuous from $\mathcal{R}(\mathcal{A})$ to $PW_{\gamma}/ \sim$.  The same results hold for $\widetilde{\mathcal{A}}$.
\end{theorem}

The proof of this is an adaptation of a similar result found in \cite{MW15a}; we include the argument in the Appendix (Subsection \ref{ssec:continuity}).  The authors of \cite{MW15a} note that in their numerical experiments, the reconstruction is not stable.  It is proven in \cite{CCD16a} that $\mathcal{A}^{-1}$ cannot be Lipschitz continuous--and thus the reconstruction cannot be stable--because the space $PW_{\gamma}$ is infinite-dimensional.  However, see \cite{alaifari2018stable} where stability can be obtained by relaxing the phase retrieval problem to allow for multiple unknown phases.

\subsection{Conjugate Phase Retrieval Using Structured Convolutions} \label{ssec:conv}

We will design a sampling scheme to solve the conjugate phase retrieval problem in $PW_{\gamma}$ in a manner similar to the scheme in Equation (\ref{Eq:struct-conv}).  To do so, we consider the conjugate phase retrieval problem in finite dimensions.  For the remainder of this section, we will consider the case of $PW_{\pi}$; all of our results can be extended to $PW_{\gamma}$ using variable substitutions.
\begin{definition}
The vectors $\{ \vec{v}_{1}, \dots, \vec{v}_{n} \} \subset \mathbb{C}^{K}$ do \emph{conjugate phase retrieval} if   
\[
| \langle \vec{x} , \vec{v}_{j} \rangle | = | \langle \vec{y}, \vec{v}_{j} \rangle | \qquad (j=1,...,n)
\]
for  $\vec{x}, \vec{y} \in \mathbb{C}^{K}$  implies that $\vec{x} = e^{i \theta} \vec{y} \text{ or } \vec{x}  = e^{i \theta} \overline{ \vec{y} }$  for some $\theta \in \mathbb{R}$.

If we write the vectors $\vec{v}_{j}$ as column vectors, we will say that the matrix $V = \begin{bmatrix} \vec{v}_{1} & \dots & \vec{v}_{n} \end{bmatrix}$ does conjugate phase retrieval when the columns of $V$ do conjugate phase retrieval.
\end{definition}



For vectors $\vec{v}, \vec{b}  \in \mathbb{C}^{K}$ with $\vec{v} = (v_0,...,v_{K-1})$ and $f \in PW_{\pi}$, we define
\begin{equation} \label{Eq:ast}
 \vec{v} \ast f = \sum_{k=0}^{K-1} \overline{v_{k}} f( \cdot + b_{k} ).
\end{equation}
We refer to this as a \emph{structured convolution}.   We can think of the sum in Equation \eqref{Eq:ast} as the inner-product of $\vec{v}$ and the vector $\left( f(\cdot + b_{0}), \dots, f(\cdot + b_{K-1}) \right)^{T}$.

For $b_0,...,b_{K1}\in\mathbb{R}$ we denote by $\mathbb{Z}( b_{0}, b_{1}, \dots, b_{K-1} )$ the subgroup of $\mathbb{R}$ generated by the integer multiples of the $b_k$. We recall that the Beurling density of $X\subseteq \mathbb{R}$ is defined by
\[
\mathcal{D}(X) = \lim_{h\to \infty} \inf_{x\in \mathbb{R}} \frac{\#(X\cap [x,x+h])}{h}
\]
if this limit exists. (We only deal with situations where Beurling's lower and upper densities coincide.)

\begin{theorem} \label{Th:main1}
Let $V = \begin{bmatrix} \vec{v}_{0} & \dots & \vec{v}_{M-1} \end{bmatrix}$ be a $K \times M$ matrix which does conjugate phase retrieval on $\mathbb{C}^{K}$.  Let $\{ b_{k} \}_{k=0}^{K-1} \subset \mathbb{R}$ be such that the group $\mathbb{Z}( b_{0}, b_{1}, \dots, b_{K-1})$ has Beurling density greater than one.  Suppose $\{ t_{n} \}_{n \in \mathbb{Z}} \subset \mathbb{R}$ is a set of sampling for the space $PW_{2 \pi}$.  Then the following sampling scheme does conjugate phase retrieval on $PW_{\pi}$:
\[ \{ | \vec{v}_{m} \ast f ( t_{n} ) | : m = 0 , 1, \dots, M-1; \ n \in \mathbb{Z} \}. \]
\end{theorem}

\begin{proof}
Suppose $f,g \in PW_{\pi}$ is such that 
\begin{equation}
| \vec{v}_{m} \ast f ( t_{n} ) | = | \vec{v}_{m} \ast g ( t_{n} ) |, \text{ for } m = 0 , 1, \dots, M-1; \ n \in \mathbb{Z}.
\end{equation}

Since $\{ t_{n} \}$ is a set of sampling for $PW_{2 \pi}$ and $|\vec{v} \ast f|^2, |\vec{v} \ast g|^2 \in PW_{2 \pi}$, we have that
\[ | \vec{v}_{m} \ast f ( x ) | = | \vec{v}_{m} \ast g (x) |, \text{ for all } x \in \mathbb{R}. \]
Since the matrix $V$ does conjugate phase retrieval in $\mathbb{C}^{K}$, for all $x \in \mathbb{R}$ we have that either
\begin{equation} \label{Eq:cpr-1}  
\begin{pmatrix} f(x + b_{0}) \\ f(x + b_{1}) \\ \vdots \\ f(x + b_{K-1}) \end{pmatrix} = \lambda_{1}(x) \begin{pmatrix} g(x + b_{0}) \\ g(x + b_{1}) \\ \vdots \\ g(x + b_{K-1}) \end{pmatrix}
\end{equation}
or
\begin{equation} \label{Eq:cpr-2}
\begin{pmatrix} f(x + b_{0}) \\ f(x + b_{1}) \\ \vdots \\ f(x + b_{K-1}) \end{pmatrix} = \lambda_{2}(x) \begin{pmatrix} \overline{g(x + b_{0})} \\ \overline{g(x + b_{1})} \\ \vdots \\ \overline{g(x + b_{K-1})} \end{pmatrix}
\end{equation}
for some $\lambda_{j}(x) \in \mathbb{C}$ with $| \lambda_{j}(x) | = 1$.  

For every $k=1, \dots, K-1$ and every $x$ such that Equation \eqref{Eq:cpr-1} holds, we have that
\begin{equation*}
 | f(x + b_{k}) - f(x)| = | \lambda_{1}(x) g(x + b_{k}) - \lambda_{1}(x) g(x)| = | g(x + b_{k}) - g(x) |.
\end{equation*}
Similarly, for $x$ such that Equation \eqref{Eq:cpr-2} holds, we have that
\begin{equation*} 
| f(x + b_{k}) - f(x) | = | \lambda_{2}(x) \overline{ g(x + b_{k}) } - \lambda_{2}(x) \overline{ g(x) } | = | g(x + b_{k}) - g(x) |.
\end{equation*}

Therefore, we have that $|f(x)| = |g(x)|$ and $| f(x + b_{k}) - f(x)| = | g(x + b_{k}) - g(x) |$ hold for all $k=1, \dots, K-1$ and all $x \in \mathbb{R}$.  By the proof of Theorem \ref{Th:McD} contained in \cite{McD04a}, we obtain that there exists a meromorphic function $W$, unimodular on $\mathbb{R}$, such that either $f = W g$ or $f = W g^{\sharp}$.  Moreover, $W$ is periodic with period $b_{k}$ for every $k=1, \dots, K-1$.  For the moment, suppose $f = W g$.  For every $x \in \mathbb{Z}(b_{0}, b_{1}, \dots, b_{K-1})$, we have $f(x) = W(x) g(x) = W(0) g(x)$, so $f$ and $W(0)g$ agree on a set with Beurling density greater than 1.  It follows that $f = W(0)g$ everywhere.  The same conclusion holds if $f = W g^{\sharp}$.
\end{proof}

If we assume in the proof of Theorem \ref{Th:main1} that the matrix $V$ does phase retrieval on $\mathbb{C}^{K}$, then only Equation (\ref{Eq:cpr-1}) can hold. This gives an analogous result for phase retrieval:

\begin{theorem} \label{Th:main2}
Let $V = \begin{bmatrix} \vec{v}_{0} & \dots & \vec{v}_{M-1} \end{bmatrix}$ be a $K \times M$ matrix which does phase retrieval on $\mathbb{C}^{K}$.  Let $\{ b_{k} \}_{k=0}^{K-1} \subset \mathbb{R}$ be such that the group $\mathbb{Z}( b_{0}, b_{1}, \dots, b_{K-1} )$ has Beurling density greater than one.  Suppose $\{ t_{n} \}_{n \in \mathbb{Z}} \subset \mathbb{R}$ is a set of sampling for the space $PW_{2 \pi}$.  Then the following sampling scheme does phase retrieval on $PW_{\pi}$:
\[ \{ | \vec{v}_{m} \ast f ( t_{n} ) | : m = 0 , 1, \dots, M-1; \ n \in \mathbb{Z} \}. \]
\end{theorem}

For the proof, we need the following elementary lemmas.

\begin{lemma} \label{L:resample}
Suppose $f$ is entire and nonzero.  The set of all $\beta \in  \mathbb{R}$  such that 
\[
 f\left( \frac{n}{B} + b_{k} - \beta\right) = 0, 
\] 
for some $k=0,\dots,K-1$  and for some $ n \in \mathbb{Z}$ has no limit point in $\mathbb{R}$.
\end{lemma}


\begin{lemma} \label{L:colinear}
Suppose $g$ is an entire function.  For fixed $\{b_0, \dots, b_{K-1} \} \subset \mathbb{R}$, the set of $x \in \mathbb{R}$ for which the vectors 
\[ \begin{pmatrix} g(x + b_{0}) \\ g(x + b_{1}) \\ \vdots \\ g(x + b_{K-1}) \end{pmatrix} \text{ and } \begin{pmatrix} \overline{g(x + b_{0})} \\ \overline{g(x + b_{1})} \\ \vdots \\ \overline{g(x + b_{K-1})} \end{pmatrix} \]
are colinear is either  $\mathbb{R}$ or has no limit point in $\mathbb{R}$.  If the set is  $\mathbb{R}$, then there exists a meromorphic function $W$ which is unimodular on $\mathbb{R}$ and periodic by the group $\mathbb{Z}(b_0, \dots, b_{K-1})$ such that $g^{\sharp} = W g$.
\end{lemma}

\begin{proof}
Fix $k \in \{1, \dots, K-1\}$.  For any $x \in \mathbb{R}$ such that the vectors are colinear,  we have
\begin{equation} \label{Eq:2x2det}
\det \begin{bmatrix} g(x + b_{0}) & \overline{g(x + b_{0})} \\ g(x + b_{k}) & \overline{g(x + b_{k})} \end{bmatrix} =0.
\end{equation}

 This can be expressed as 
\[
g(x + b_{0}) g^{\sharp}(x + b_{k}) - g(x + b_{k}) g^{\sharp}(x + b_{0}) = 0.
\]

As the left hand side is the restriction of an entire function to the real axis, it is either zero everywhere, or zero only on a set without limit points.  It follows that if the vectors are not colinear for all $x \in \mathbb{R}$, then they can be colinear only on a set that has no limit points.

If the vectors are colinear everywhere, then \eqref{Eq:2x2det} holds for all $x$ and for all $k \in \{0, \dots, K-1\}$.  Then  \cite[Theorem 1]{McD04a} guarantees the existence of a meromorphic $W$ such that $g^{\sharp} = Wg$.  For every $k$ and every $x$, we have $g^{\sharp}( x + b_{k} ) = W(x + b_{k} ) g(x + b_{k})$ and also $g^{\sharp}(x + b_{k}) = W(x) g(x + b_{k})$ by the colinearity assumption.  Therefore, $W$ is periodic with period $b_{k}$ for $k = 0, \dots, K-1$.
\end{proof}

\begin{proof}[Proof of Theorem \ref{Th:main2}.]
Suppose $f,g \in PW_{\pi}$ is such that 
\[ | \vec{v}_{m} \ast f(t_{n} ) | = | \vec{v}_{m} \ast g(t_{n}) | \text{ for all } m = 0, \dots, K-1, \text{ and } n \in \mathbb{Z}. \]  
Since the hypotheses of Theorem \ref{Th:main1} hold, we have that either $f = \lambda g$ or $f = \lambda g^{\sharp}$.  Suppose $f = \lambda g^{\sharp}$.  Then again by Theorem \ref{Th:main1}, we have that
\begin{equation*}
| \vec{v}_{m} \ast f(x) | = | \vec{v}_{m} \ast g^{\sharp}(x)| \text{ for all } m = 0 , \dots, K_1, \text{ and } x \in \mathbb{R}.
\end{equation*}
By our assumption, we then also have
\begin{equation*}
| \vec{v}_{m} \ast g(x) | = | \vec{v}_{m} \ast g^{\sharp}(x)| \text{ for all } m = 0 , \dots, K_1, \text{ and } x \in \mathbb{R}.
\end{equation*}
Since the matrix $A$ does phase retrieval on $\mathbb{C}^{K}$, we must have that the vectors
\begin{equation*}
\begin{pmatrix} g(x + b_{0}) \\ g(x + b_{1}) \\ \vdots \\ g(x + b_{K-1}) \end{pmatrix} \text{ and } \begin{pmatrix} \overline{g(x + b_{0})} \\ \overline{g(x + b_{1})} \\ \vdots \\ \overline{g(x + b_{K-1})} \end{pmatrix} 
\end{equation*}
are colinear for every $x \in \mathbb{R}$.  Consequently, by Lemma \ref{L:colinear}, $g = W g^{\sharp}$ for $W$ meromorphic, and periodic by the group $\mathbb{Z}( b_{0}, b_{1}, \dots, b_{K-1} )$.  As in the proof of Theorem \ref{Th:main1}, since $g, g^{\sharp} \in PW_{\pi}$, $W$ is constant.  Hence, we have $f = \lambda g^{\sharp} = (\mu \lambda) g$ with $| \mu | = 1$.  This completes the proof.
\end{proof}

To demonstrate that Theorem \ref{Th:main1} is not vacuous, we present an example here.  For this purpose, we require a result from \cite{EL2017} concerning conjugate phase retrieval in $\mathbb{C}^{2}$ and $\mathbb{C}^{3}$:

\begin{thmalpha} \label{Th:cpr}
If $\vec{v}_{1}, \vec{v}_{2}, \vec{v}_{3} \in \mathbb{R}^2$ is written as
\[ \begin{bmatrix} \vec{v}_{1} & \vec{v}_{2} &\vec{v}_{3} \end{bmatrix}  = \begin{bmatrix} a_1 & b_1 & c_1 \\ a_2 & b_2 & c_2 \end{bmatrix} \]
then $\vec{v}_{1}, \vec{v}_{2}, \vec{v}_{3}$ does conjugate phase retrieval in $\mathbb{C}^{2}$ if and only if 
\begin{equation}\label{eq:det2} \det\left[ \begin{array}{ccc} a_1^2 & 2a_1a_2 &
a_2^2 \\ b_1^2 & 2b_1b_2 & b_2^2 \\ c_1^2 & 2c_1c_2 & c_2^2 \end{array}
\right]\ne 0.  \end{equation}

Likewise, if $\vec{v}_{1}, \dots, \vec{v}_{6} \in \mathbb{R}^{3}$ is written as
\[ \begin{bmatrix} \vec{v}_{1} & \vec{v}_{2} & \dots & \vec{v}_{6} \end{bmatrix}  =
\begin{bmatrix} a_1 & b_1
& c_1 & d_1 & e_1&f_1 \\ a_2 & b_2 & c_2 & d_2 & e_2&f_2\\ a_3 & b_3& c_3& d_3 &
e_3&f_3\\ \end{bmatrix}  \]
then $\vec{v}_{1}, \dots \vec{v}_{6}$ does conjugate phase retrieval in $\mathbb{C}^{3}$ if and only if
\begin{equation}\label{eq:det3} \det\left[
\begin{array}{cccccc} a_1^2 & a_2^2  & a_3^2 & 2a_1a_2 & 2a_1a_3& 2a_2a_3\\
b_1^2 &b_2^2  & b_3^2 & 2b_1b_2 & 2b_1b_3&2b_2b_3\\ c_1^2 & c_2^2  & c_3^2 &
2c_1c_2 & 2c_1c_3&2c_2c_3\\ d_1^2 &d_2^2  & d_3^2 & 2d_1d_2 & 2d_1d_3&2d_2d_3\\
e_1^2 & e_2^2  & e_3^2 & 2e_1e_2 & 2e_1e_3&2e_2e_3\\ f_1^2 & f_2^2  & f_3^2 &
2f_1f_2 & 2f_1f_3&2f_2f_3\\ \end{array} \right]\ne 0.  \end{equation}
\end{thmalpha}

\begin{example} \label{Ex:V}
It is easy to check that the following matrix does conjugate phase retrieval on $\mathbb{C}^{3}$ using Theorem \ref{Th:cpr}:
\begin{equation} \label{Eq:cpr-matrix}
 V = \left[ \begin{array}{rrrrrr}  1 & 0 & 0 & 1 & 1 & 0 \\ 0 & 1 & 0 & -1 & 0 & 1 \\ 0 & 0 & 1 & 0 & -1 & -1 \end{array} \right].
\end{equation}
Thus, we choose this matrix $V$, $t_{n} = \frac{n}{2}$, and $b_{0} = 0$, $b_{1} = \frac{1}{2}$, and $b_{2} = 1$.  We will demonstrate in Subsection \ref{ssec:numerics} the results of numerical experiments involving the reconstruction algorithm we propose in Subsection \ref{ssec:recon}.
\end{example}

We note that the condition on the coefficient matrix $V$ for the structured convolutions in Theorem \ref{Th:main2} is much more restrictive than in Theorem \ref{Th:main1}.  Indeed, Example \ref{Ex:V} illustrates this distinction.  We also note that Theorem \ref{Th:main2} generalizes the results in \cite{PYB14a} in the following sense:  the structured modulations used in \cite{PYB14a} correspond to $b_{k} \in \mathbb{C}$ (and in fact, some $b_{k}$ must be non-real), whereas our result applies if $b_{k}$ are real.

\subsection{Reconstruction methods} \label{ssec:recon}

The proof of Theorem \ref{Th:main1} suggests a reconstruction method.  We wish to reconstruct $f \in PW_{\pi}$ from the samples 
\begin{equation} \label{Eq:samples}
 \{ | \vec{v}_{m} \ast f ( t_{n} ) | : m = 0 , 1, \dots, M-1; \ n \in \mathbb{Z} \}
\end{equation}
where $\{ t_{n} \}$ and $\vec{v}_{m}$ satisfy the hypotheses of Theorem \ref{Th:main1}.  The strategy of Algorithm \ref{RA:conv} is to reconstruct from the samples given in Theorem \ref{Th:main1} the function values
\begin{equation*}
\{ \lambda f(x_{n}) \} \text{ or } \{ \lambda \overline{f(x_{n})} \} 
\end{equation*} 
on a sequence of points $\{ x_{n} \}$ which is a set of sampling for $PW_{\pi}$.  We will not be able to determine $\lambda$, nor will we be able to determine whether we reconstruct the function values or their conjugates; we will reconstruct them up to uniform phase factor $\lambda$ and uniform choice of conjugation.

\begin{algorithm}[t] 
\caption{Reconstruct $[f]$ from structured convolutions} \label{RA:conv}
\begin{algorithmic}[1]
\STATE Given $|\vec{v}_{m} \ast f (t_{n})|^2$ for $m=0,\dots,M-1$, $n \in \mathbb{Z}$;
\STATE Initialize $\beta$ randomly;
\STATE use Equation \ref{Eq:interpolate} to calculate \[ | \vec{v}_{m} \ast f (x_{n} - \beta) |^2, \quad m=0,1, \dots, M-1, \ n \in \mathbb{Z}; \]
\STATE apply the Gerchberg-Saxton method to calculate \begin{equation} \label{Eq:adjacent}  \vec{F}(x_{n} - \beta) := \lambda( x_{n} - \beta) \begin{pmatrix} f(x_{n} + b_{0} - \beta) \\ f(x_{n} + b_{1} - \beta ) \\ \vdots \\ f(x_{n} + b_{K-1} - \beta) \end{pmatrix} 
\end{equation}
up to the unknown phase $\lambda( x_{n} - \beta)$ and unknown conjugation;
\STATE for $n$, choose $\lambda( x_{n} - \beta)$ and conjugation in Equation \eqref{Eq:adjacent} so that 
\[ \vec{F}(x_{n-1} - \beta)  \text{ and } \vec{F}(x_{n} - \beta) \]
are consistent;
\STATE use Equation \ref{Eq:interpolate} to reconstruct either $\lambda f$ or $\lambda f^{\sharp}$ (and hence $[f]$) from
\begin{equation*}
\{ \lambda f(x_{n} - \beta): n \in \mathbb{Z} \}  \text{ or } \{ \lambda \overline{f(x_{n} - \beta)}: n \in \mathbb{Z} \} 
\end{equation*}
for our choice of uniform phase factor $\lambda$ and conjugation.
\end{algorithmic}
\end{algorithm}

For the unknown signal $f \in PW_{\pi}$, for any $x \in \mathbb{R}$, we define the vectors
\begin{equation}
\vec{F}(x) = \begin{pmatrix} f(x + b_{0}) \\ f(x + b_{1}) \\ \vdots \\ f(x + b_{K-1}) \end{pmatrix}.
\end{equation}
Using the arguments from Theorem \ref{Th:main1}, we can reconstruct for any $x$ the vector $\lambda(x) \vec{F}(x)$ or $\lambda(x) \overline{\vec{F}}(x)$ up to unknown phase factor $\lambda(x)$ and unknown conjugation.  We choose a sequence $\{x_{n} \}$ that has the following properties:
\begin{enumerate}
\item $\{ x_{n} \}$ is a set of sampling for $PW_{\pi}$;
\item for every $n$, the vectors $\vec{F}(x_{n-1})$ and $\vec{F}(x_{n})$ have at least two entries in common.  In other words, we want
\[ \# \left( \{ x_{n-1} + b_{k} : k=0, \dots, K-1\} \cap \{ x_{n} + b_{k} : k=0, \dots, K-1\} \right) \geq 2 \]
\end{enumerate}
Since the vectors $\lambda(x_{n-1}) \vec{F}(x_{n-1})$ and $\lambda(x_{n}) \vec{F}(x_{n})$ have two entries in common (say $x_{n} + b_{j}$ and $x_{n} + b_{k}$), the ambiguity of phase factor and choice of conjugation can be rectified so that they are consistent, provided the following matrix has nonzero determinant:
\begin{equation} \label{Eq:colinear2}
M(x_{n},j,k) := \begin{bmatrix} f(x_{n} + b_{j}) & \overline{f(x_{n} + b_{j})} \\  f(x_{n} + b_{k}) & \overline{f(x_{n} + b_{k})} \end{bmatrix}.
\end{equation}
By Lemma \ref{L:colinear}, we have for any choice of distinct $b_{j}$ and $b_{k}$, the set of $\{x_{n}\}$ such that the determinant of the matrix in Equation \eqref{Eq:colinear2} is $0$ is either countable or all of $\mathbb{R}$.  As we saw in the proof of Theorem \ref{Th:main2}, if the determinant is $0$ for all of $\mathbb{R}$, then $f^{\sharp} = \lambda f$ for some uniform phase factor $\lambda$, and hence, either choice of conjugation is consistent up to a phase factor.  

Suppose for the moment that $f^{\sharp}$ and $f$ are linearly independent.  Since it is still possible that some, but not all, of the determinants in Equation \eqref{Eq:colinear2} could be $0$, we choose a $\beta \in \mathbb{R}$ randomly.  We then endeavor to reconstruct $\lambda(x_{n} - \beta) \vec{F}(x_{n} - \beta)$ for $n \in \mathbb{Z}$ as before, and make successive samples consistent by considering the determinant of the matrices $M(x_{n} - \beta, j, k)$ instead.  We want $\beta$ to be chosen so that for every $n \in \mathbb{Z}$, the determinant of this matrix is nonzero.  However, we know that (since we are assuming for the moment that $f^{\sharp}$ and $f$ are linearly independent) the set of $\beta$ that fails to have this property is at most countable.  Therefore, if we choose $\beta$ randomly with respect to any continuous probability distribution on $\mathbb{R}$ (or $[0,1]$), with probability $1$ we will obtain that all of the determinants of $M(x_{n} - \beta, j, k)$ are nonzero.

To sum up, after our choice of $\beta$, with probability $1$, we will either have:
\begin{enumerate}
\item  $\det M(x_{n} - \beta, j, k) \neq 0$ for all $n \in \mathbb{Z}$, $j,k = 0, \dots, K-1$;
\item  $\det M(x_{n} - \beta, j, k) = 0$ for all $n \in \mathbb{Z}$, $j,k = 0, \dots, K-1$.
\end{enumerate}
If the condition of Item 2. holds, then we actually have that $\det M(x - \beta, j, k) = 0$ for all $x \in \mathbb{R}$.  Hence, as observed previously, either choice of conjugation is consistent between $\lambda(x_{n-1} - \beta) \vec{F}(x_{n-1} - \beta)$ and $\lambda(x_{n} - \beta) \vec{F}(x_{n} - \beta)$.  If all of the determinants are nonzero, then the choice of conjugation is uniquely determined to make the samples consistent.  Thus, we have the reconstruction method:

We note that in Step 5, we can initialize our choice of phase factor and conjugation for $\vec{F}(x_{0} - \beta)$ arbitrarily, then work outward in both directions for $n > 0$ and $n< 0$.  We will discuss the Gerchberg-Saxton method of Step 4 in more detail in Subsection \ref{ssec:numerics}.

\subsection{Specific Structured Convolutions for Conjugate Phase Retrieval} \label{ssec:cpr}

We demonstrate here a sampling scheme using simple structured convolutions and the corresponding reconstruction as outlined in Algorithm \ref{RA:conv} to do conjugate phase retrieval in $PW_{\pi}$.  For convenience, we structure the convolutions so that the $b_{k} = \frac{k}{B}$, for some integer $B > 1$, $K \geq 3$, and $t_{n} = \frac{n}{B}$.  In particular, if we choose $B = 2$, then we can use the coefficient matrix $V$ as given in Equation \eqref{Eq:cpr-matrix}.  Note that $\mathbb{Z}(b_{0}, b_{1}, b_{2}) = \frac{1}{2} \mathbb{Z}$, and so satisfies the conditions of Theorem \ref{Th:main1}.

With the lattice structure of the sampling points $\{ t_{n}\}$ and the $b_{k}$'s also lying on the same lattice, we obtain that the samples $| \vec{v}_{2} \ast f(t_{n}) |$ and $| \vec{v}_{3} \ast f(t_{n}) |$ are repetitions of the samples $| \vec{v}_{1} \ast f(t_{n}) |$.  Likewise, the samples $| \vec{v}_{6} \ast f(t_{n}) |$ are repetitions of the samples $| \vec{v}_{4} \ast f(t_{n}) |$.  Thus, we only need to sample the functions $| \vec{v}_{1} \ast f|$, $| \vec{v}_{4} \ast f |$, and $| \vec{v}_{5} \ast f |$.

We note that this sampling scheme requires sampling $3$ functions at twice the Nyquist rate, and thus our oversampling factor is $6$.  We can reduce this down to oversampling by a factor of $3$ by incorporating our choice of $\beta$ into the sampling scheme:
\begin{algorithm} 
\caption{Reconstruct $[f]$ from samples at 3 times the Nyquist rate} \label{R:3x-nyquist}
\begin{algorithmic}[1]
\STATE Choose $\beta$ at random.
\STATE Sample $| \vec{v}_{m} \ast f( n - \beta) |$ for $m =1,4,5$ and $n \in \mathbb{Z}$.
\STATE For each $n$, use the samples in Step 2 to reconstruct the vector
\[ \vec{F}(n) = \lambda(n) \begin{pmatrix} f( n+1 - \beta) \\ f(n - \beta) \\ f(n - 1 - \beta) \end{pmatrix} \]
up to unknown phase $\lambda(n)$ and unknown conjugation.
\STATE Choose the phase and conjugation for $\vec{F}(n+1)$ from the choice of phase and conjugation for $\vec{F}(n)$, since they have 2 entries that coincide.
\end{algorithmic}
\end{algorithm}

This algorithm will not work on all signals in $PW_{\pi}$, but only on \emph{generic} signals $f \in PW_{\pi}$.  By this we mean that the set of signals for which this algorithm fails is meager (i.e. of First Category).  
\begin{lemma}
For a fixed $\beta \in \mathbb{R}$, the set of signals $f \in PW_{\pi}$ for which Reconstruction Algorithm \ref{R:3x-nyquist} fails is meager.
\end{lemma}
\begin{proof}
Reconstruction Algorithm \ref{R:3x-nyquist} will reconstruct $f$ up to unknown phase and conjugation whenever for all $n \in \mathbb{Z}$, $f(n - \beta) \neq 0$ and $f(n+1 - \beta) \overline{f(n - \beta)} - \overline{f(n+1 - \beta)} f(n - \beta) \neq 0$.  Clearly the latter condition implies the former.  Therefore, if we consider the set 
\[ \mathcal{F}_{n} := \{ f \in PW_{\pi} : f(n+1 - \beta) \overline{f(n - \beta)} - \overline{f(n+1 - \beta)} f(n - \beta) = 0 \}, \]
we see that the complement of $\mathcal{F}_{n}$ is open and dense.  The lemma follows since $\cup_{n} \mathcal{F}_{n}$ is the set of signals for which the reconstruction fails.
\end{proof}

It is known that in $\mathbb{C}^{K}$, a frame must have at least $4K - 4$ vectors in order to do phase retrieval \cite{BCE06a}.  No such bound is known for conjugate phase retrieval, but note that our sampling scheme above suggests that it should be on the order of $3K$.


\subsection{Conjugate Phase Retrieval using Derivatives} \label{ssec:deriv}

In analogy to structured convolutions, conjugate phase retrieval is possible by sampling the derivative of the unknown signal.  Reconstruction of a signal from samples of its derivatives is gaining interest \cite{grochenig2018sampling,goncalves2017interpolation,goncalves2018interpolation}.

\begin{lemma} \label{L:fg-unique}
Suppose $f$ and $g$ are entire functions with the property that $ff^{\sharp} = g g^{\sharp}$ and $f' f'^{\sharp} = g' g'^{\sharp}$.  Then there exists a unimodular scalar $\lambda$ such that either $f = \lambda g$ or $f = \lambda g^{\sharp}$.
\end{lemma}

This is a restatement of part of Theorem \ref{Th:obvious}.
\begin{theorem} \label{Th:main}
Suppose $\{ t_{n} \}$ is a set of sampling for $PW_{2 \pi}$.  Then the mapping 
\begin{equation*}
\widetilde{\mathcal{A}} : PW_{\pi}/ \sim \to \ell^{2}(\mathbb{Z}) \oplus \ell^{2}(\mathbb{Z}) : f \mapsto ( | f(t_{n})|, | f^{\prime}(t_{n})|)_{n}
\end{equation*}
is one-to-one.
\end{theorem}

We write
\begin{equation} \label{Eq:fg-phase}
f(t) = r(t) e^{i \theta(t)} \quad t \in \mathbb{R}, \ r(t) \geq 0, \ \theta(t) \in \mathbb{R}.
\end{equation}
The functions $r,\theta$ are differentiable a.e.  Theorem \ref{Th:main} and Lemma \ref{L:fg-unique} provide a theoretical reconstruction algorithm as follows.

\begin{algorithm} 
\caption{Reconstruct $[f]$ from derivative sampling} \label{RA:derivative}
\begin{algorithmic}[1]
\STATE Given the phaseless samples $\{  |f(t_{n})|, | f'(t_{n})| \}$
\STATE reconstruct $f f^{\sharp}$ and $f' f'^{\sharp}$ in $PW_{2 \pi}$;
\STATE reconstruct $r = \sqrt{ f f^{\sharp} }$;
\STATE reconstruct
\[  ( \theta^{\prime} )^{2} = \dfrac{ f' f'^{\sharp} }{ f f^{\sharp} } - \dfrac{ [ ( f f^{\sharp} )^{\prime} ]^{2} }{ 4 (f f^{\sharp})^2} \]
on some interval $I$;
\STATE choose a square-root of $( \theta^{\prime} )^2$ and integrate;
\STATE use $f = r e^{i \theta}$ on $I$ to expand $f$ as a power series.
\end{algorithmic}
\end{algorithm}
Of course, this cannot be reasonably done numerically.

\section{Numerical Methods and Experiments}

In this section we will describe our implementation of Algorithm \ref{RA:conv} and the results of numerical experiments.  Recall that in Step 5 of Algorithm \ref{RA:conv}, we used the Gerchberg-Saxton method to reconstruct $\vec{F}(x_{n} - \beta)$ (which has phase information about the samples of the unknown signal $f$) from the unphased samples $\{ | \vec{v}_{m} \ast f (x_{n} - \beta )| : m=0,\dots,M-1\}$.  We first consider the results of our numerical experiments of this method for conjugate phase retrieval in $\mathbb{C}^{K}$.

\subsection{The Gerchberg-Saxton Method of Alternating Projections} \label{ssec:ap}
For a matrix that does phase retrieval on $\mathbb{C}^{K}$, there are many reconstruction techniques:  frame methods \cite{BCE06a,BBCE09a}; convex optimization \cite{candes2013phase,goldstein2018phasemax}; and the Kaczmarz method \cite{wei2015solving,tan2019phase} to name only a few.  However, for a matrix $V$ that does conjugate phase retrieval on $\mathbb{C}^{K}$ (but not phase retrieval), there is no known proven method of reconstruction.  None of the previously mentioned reconstruction techniques for phase retrieval extend in an obvious way to conjugate phase retrieval, because they all utilize the fact that in the space $\mathbb{C}^{K}$, there is only one linearly independent solution to the inverse problem.  In the case of conjugate phase retrieval, there are two linearly independent solutions to the inverse problem, namely the original signal and its conjugate.

Despite this shortcoming, we will demonstrate experimentally that the Gerchberg-Saxton method \cite{gerchberg1972practical} can be used for reconstruction.  Suppose $V = \begin{bmatrix} \vec{v}_{0} \dots \vec{v}_{M-1} \end{bmatrix}$ that does conjugate phase retrieval on $\mathbb{C}^{K}$.  Suppose $\vec{y} \in \mathbb{C}^{K}$; our aim is to reconstruct $[ \vec{y} ]$ from $| V^{*} \vec{y} |$.  We begin by choosing phases $\{ \alpha_{0} , \dots, \alpha_{M-1} \} \subset \mathbb{C}$, $|\alpha_{j}| = 1$ and form the initial estimate
\begin{equation*}
\vec{x}^{0} = \begin{pmatrix} \alpha_{0} | \langle \vec{y}, \vec{v}_{0} \rangle | \\ \vdots \\ \alpha_{M-1} | \langle \vec{y}, \vec{v}_{M-1} \rangle |\end{pmatrix}.
\end{equation*}
We let $V^{\dagger}$ be the Moore-Penrose inverse of $V^{*}$ (note that $V^{*}$ must be injective for $V$ to do conjugate phase retrieval), and we define $S : \mathbb{C}^{M} \to \mathbb{C}^{M}$ to be the nonlinear projection onto the set:
\begin{equation*} 
\{ \vec{w} \in \mathbb{C}^{M} : |\vec{w}| = |V^{*} \vec{y} | \}.
\end{equation*}
Following the Gerchberg-Saxton method of alternating projections, we define the sequence of estimates $\vec{x}^{n}$ by:
\begin{equation} \label{Eq:alt-proj}
\vec{x}^{n+1} = S V^{*} V^{\dagger} \vec{x}^{n}.
\end{equation}

Levi and Stark \cite{LS84a} prove that under a particular metric on $\mathbb{C}^{K}/\sim$, the sequence of estimates given in Equation \eqref{Eq:alt-proj} converges.  However, they do not prove that the sequence converges to the desired solution, and in fact demonstrate the that alternating projections method can become stuck in what they refer to as ``traps'' and ``tunnels''.  Recently, \cite{netrapalli2013phase} prove that the sequence of estimates in Equation \eqref{Eq:alt-proj} converges to the solution in expectation provided the matrix $V$ is a Gaussian ensemble.  Our matrix $V$ as in Equation \eqref{Eq:cpr-matrix} does not satisfy this condition, however.

We performed the Gerchberg-Saxton method given in Equation \eqref{Eq:alt-proj} on 1,000 instances of vectors in $\mathbb{C}^{3}$ using the matrix $V$ as in Equation \eqref{Eq:cpr-matrix}.  Each instance of input vector $\vec{y} \in \mathbb{C}^{3}$ was generated using the \texttt{rand} function in MATLAB.   For each instance of initial vector $\vec{y}$, we chose initial phases $\alpha_{0}, \dots, \alpha_{5}$ also using the \texttt{rand} function.  For each instance, we ran 900 iterations of the alternating projections.  For each instance $\vec{y}$ and each iteration $\vec{x}^{n}$, we calculated the reconstruction error
\begin{equation*}
\varepsilon_{n}(\vec{y}) = \min \{ \| \vec{y} \otimes \vec{y} - (V^{\dagger} \vec{x}^{n}) \otimes (V^{\dagger} \vec{x}^{n}) \|_{F}, \| \vec{y} \otimes \vec{y} - (V^{\dagger} \overline{\vec{x}^{n}} ) \otimes (V^{\dagger} \overline{\vec{x}^{n}}) \|_{F} \}
\end{equation*}
where $\| \cdot \|_{F}$ is the Fr\"obenius norm, and $\overline{ \vec{z} }$ is the vector with conjugated entries.  Our target was a reconstruction error satisfying $ \varepsilon_{n}(\vec{y}) < 10^{-8}$.  We counted the number of instances for which $\varepsilon_{900}(\vec{y}) < 10^{-8}$, and of those instances, the mean and median $n$ to obtain $\varepsilon_{n}(\vec{y}) < 10^{-8}$.  Our results are tabulated in Table \ref{Tab:iterations}.

\begin{table}[!h]
\begin{tabular}{| l | r |}
\hline
\# successful reconstructions & 850 (85\%) \\
\hline
mean \# iterations to threshhold & 185.97 \\
\hline
median \# iterations to threshhold & 124 \\
\hline
\end{tabular}
\vspace{0.5ex}
\caption{Experimental results of the Gerchberg-Saxton Method.}
\label{Tab:iterations}
\end{table}



The instances for which the reconstruction was not successful, meaning $\varepsilon_{900}(\vec{y}) \geq 10^{-8}$, illustrate the traps and tunnels phenomenon observed by Levi and Stark.  As the reconstruction errors are decreasing in general (but not always; see Figure \ref{fig:alt-proj} (B)), those instances whose errors are not converging to $0$ exhibit these phenomena.

\begin{figure}[h] 
\hspace{-1cm}
\subfloat[]{\includegraphics[width=6cm]{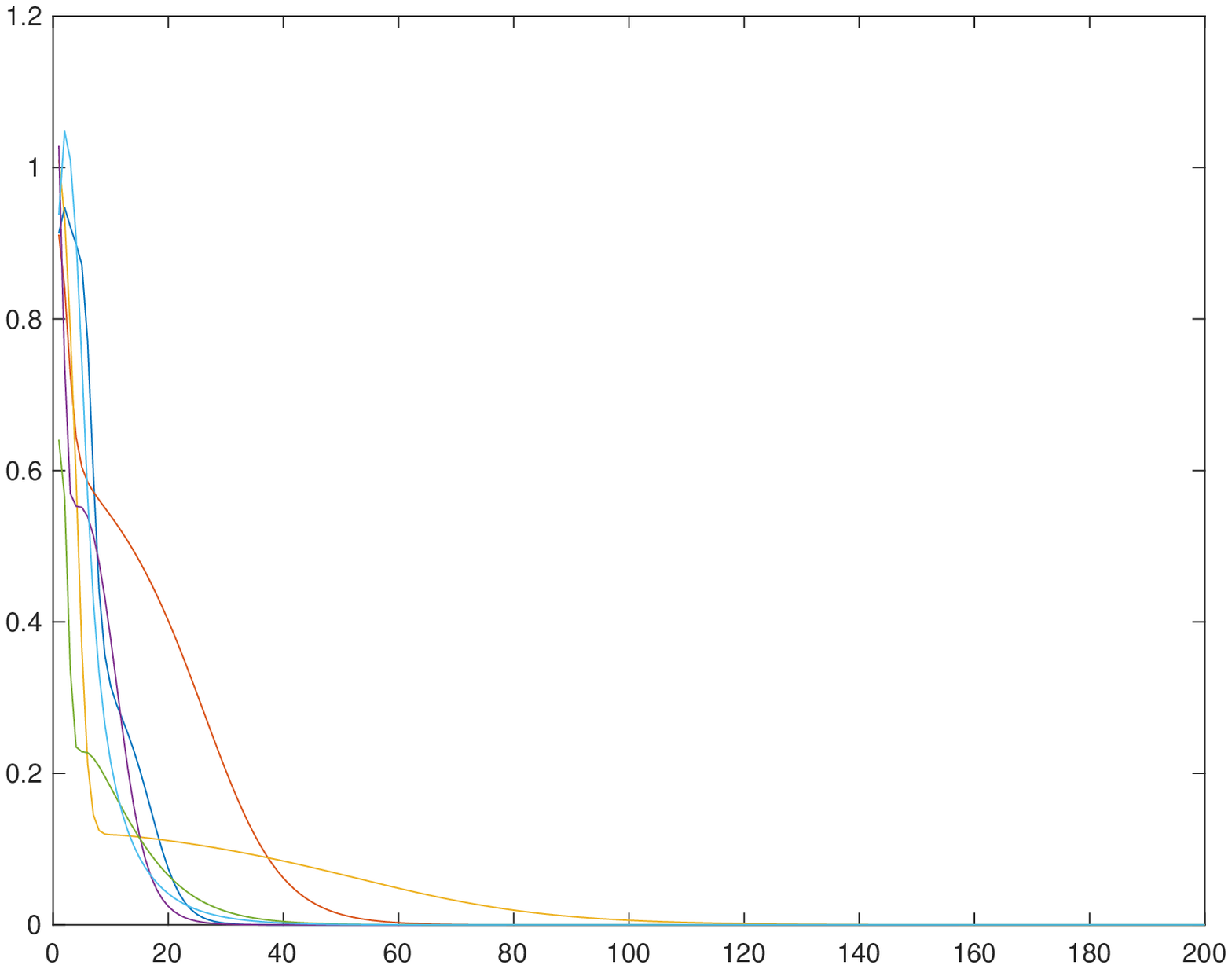}}
\subfloat[]{\includegraphics[width=6cm]{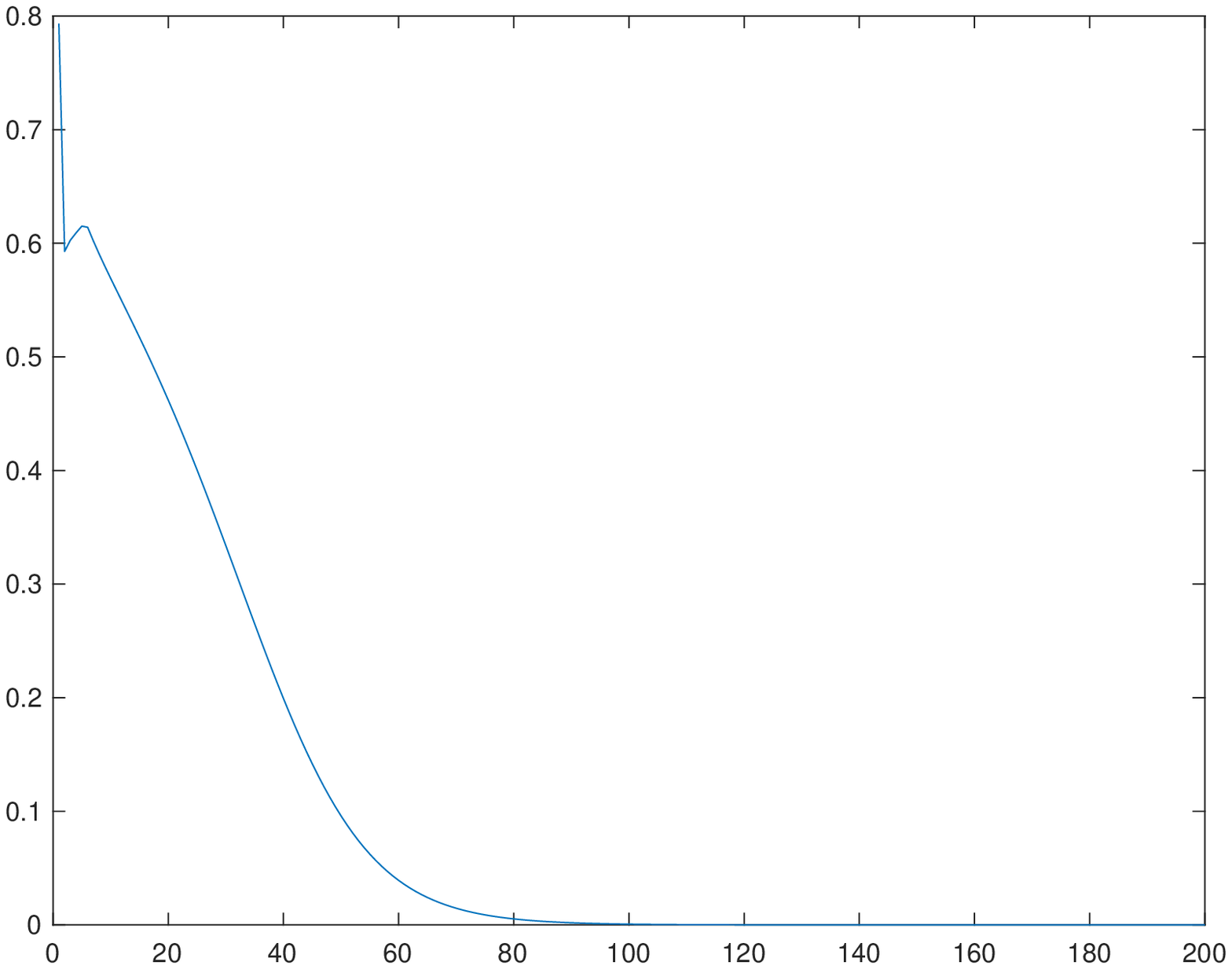}}
\subfloat[]{\includegraphics[width=6cm]{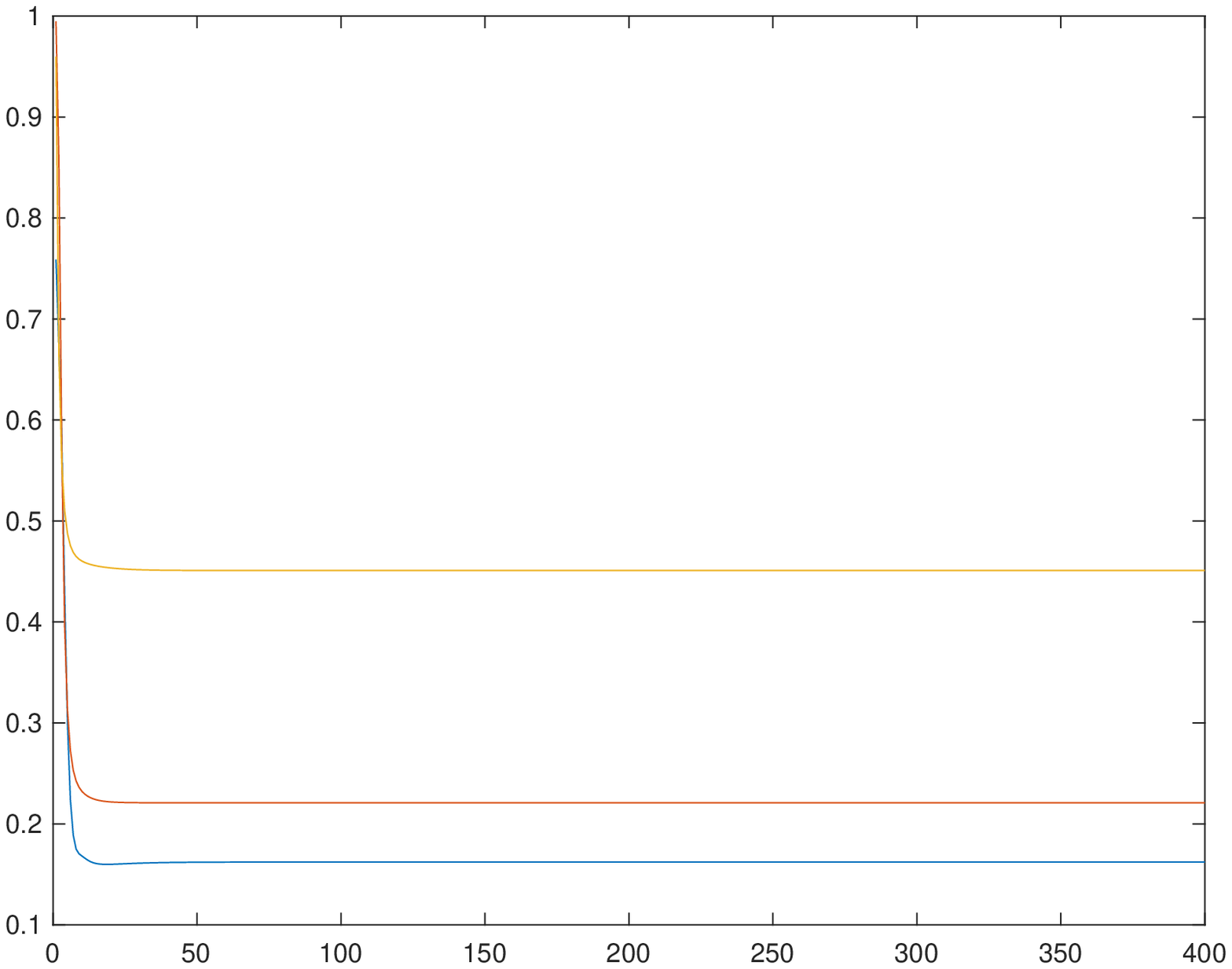}}
\caption{(A) typical error decay of successful reconstructions; (B) reconstruction with several iterations of increasing error; (C) typical error decay of unsuccessful reconstructions.}
\label{fig:alt-proj}
\end{figure}

\subsection{Implementation of Algorithm \ref{RA:conv}} \label{ssec:numerics}

We performed several numerical experiments to instantiate Algorithm \ref{RA:conv}.  Using Theorem \ref{Th:main1} and the discussion leading up to Algorithm \ref{RA:conv}, we chose the following parameters: 
\[ t_{n} = x_{n} = \dfrac{n}{2}; \qquad b_0 = 0, \ b_{1} = \dfrac{1}{2}, \text{ and } b_{2} = 1, \] 
and the coefficient matrix $V$ as in Equation \eqref{Eq:cpr-matrix}.
To begin, we defined a function $f \in PW_{\pi}$ by randomly generating complex numbers representing the samples $f(-10), \dots, f(-1)$ and $f(1), \dots, f(10)$, and set $f(0) = 0$.  Thus, our signal is 
\begin{equation} \label{Eq:interpolation}
f(t) = \sum_{n=-10}^{10} f(n) \ \text{sinc}(t - n).
\end{equation}
Using these samples, we populated the array
\begin{equation}
P = \begin{bmatrix} \cdots & f( \frac{n+1}{2} ) & f( \frac{n+2}{2} ) & \cdots \\ \cdots & f( \frac{n}{2} ) & f( \frac{n+1}{2} ) & \cdots \\ \cdots & f(\frac{n-1}{2}) & f(\frac{n}{2}) & \cdots \end{bmatrix}
\end{equation}
for $n=-40,\dots,40$ via the interpolation formula in Equation \eqref{Eq:interpolation}.  We then used 
\begin{equation}
R = | V^{T}  P |
\end{equation}
as the input data to Algorithm \ref{RA:conv}.

We chose $\beta$ randomly using MATLAB's \texttt{rand} function.  Using the entries of $R$ and the interpolation formula \eqref{Eq:interpolation2} truncated to $n=-40,\dots,40$, we reconstructed the entries of the matrix

\begin{equation} \label{Eq:unphased}
R_{\beta} = \begin{bmatrix}
\cdots & |f( \frac{n+1}{2} - \beta)| & |f( \frac{n+2}{2} - \beta)|  & \cdots \\ \cdots & | f(\frac{n}{2} - \beta)|  & | f(\frac{n+1}{2} - \beta)|   & \cdots \\ \cdots & | f(\frac{n-1}{2} - \beta) | & | f(\frac{n}{2} - \beta) |  & \cdots  \\  \cdots & |f( \frac{n+1}{2} - \beta) - f(\frac{n}{2} - \beta) | & |f( \frac{n+2}{2} - \beta) - f(\frac{n+1}{2} - \beta) |  & \cdots  \\ \cdots &  |f( \frac{n+1}{2} - \beta) - f(\frac{n-1}{2} - \beta) |  &  |f( \frac{n+2}{2} - \beta) - f(\frac{n}{2} - \beta) |  & \cdots  \\ \cdots &  | f(\frac{n}{2} - \beta) - f(\frac{n-1}{2} - \beta) | &  | f(\frac{n+1}{2} - \beta) - f(\frac{n}{2} - \beta) |  & \cdots 
\end{bmatrix}
\end{equation}

Since $|f|^2 \in PW_{2 \pi}$, we have
\begin{equation} \label{Eq:interpolation2}
|f(t)|^2 = \sum_{n \in \mathbb{Z}} \left| f \left( \dfrac{n}{2} \right) \right|^2 \text{sinc}(2t - n).
\end{equation}
Similarly for $|f(\cdot + \frac{1}{2}) - f(\cdot)|^2$ and $|f(\cdot + \frac{1}{2}) - f(\cdot - \frac{1}{2})|^2$.

We applied the Gerchberg-Saxton method of alternating projections as described in Subsection \ref{ssec:ap} to each column of the matrix in \eqref{Eq:unphased} to obtain the estimate $\lambda(\frac{n}{2} - \beta) \vec{F}(\frac{n}{2} - \beta)$ of the $n$th column of the matrix
\begin{equation*}
P_{\beta} = \begin{bmatrix} \cdots & f( \frac{n+1}{2} - \beta ) & f( \frac{n+2}{2} - \beta) & \cdots \\ \cdots & f( \frac{n}{2} - \beta) & f( \frac{n+1}{2} - \beta) & \cdots \\ \cdots & f(\frac{n-1}{2} -\beta) & f(\frac{n}{2} - \beta) & \cdots \end{bmatrix}.
\end{equation*}
As we observed in Subsection \ref{ssec:ap}, the Gerchberg-Saxton method can fail (e.g. Figure \ref{fig:alt-proj} (C)), so we apply the method to each column of the matrix 100 times, each with a different (random) seeding.  For $k= 1,\dots, 100$, we utilize 900 iterations of Equation \eqref{Eq:alt-proj} to obtain an estimate $\vec{x}_{k}$, then choose
\begin{equation*} 
\lambda(\frac{n}{2} - \beta) \vec{F}(\frac{n}{2} - \beta) = V^{\dagger} \left( \text{argmin }_{k=1,\dots,100} \| \vec{y}_{n} - | \vec{x}_{k}| \| \right)
\end{equation*}
where $\vec{y}_{n}$ is the $n$th column of $R_{\beta}$.

This estimate will be ambiguous up to unknown phase factor and conjugation.  As the $n$-th and $n+1$-st column of $P_{\beta}$ have two entries in common, working from $n=-40$ to $n=40$, we choose $\lambda(\frac{n}{2} - \beta)$ and, if necessary, conjugate $\vec{F}(\frac{n}{2} - \beta)$, so that the corresponding entries of $\vec{F}(\frac{n-1}{2})$ and $\lambda(\frac{n}{2} - \beta) \vec{F}(\frac{n}{2} - \beta)$ agree.  With probability $1$, there is no ambiguity in these choices by Lemmas \ref{L:resample} and \ref{L:colinear}.

We ran the Algorithm \ref{RA:conv} on 100 instances of $f \in PW_{\pi}$ with randomly generated values for $f(-10), \dots, f(-1), f(1), \dots, f(10)$.  For each instance, we ran Algorithm \ref{RA:conv} 20 times, each with a different (randomly generated) value of $\beta$.  We then chose the reconstruction $r$ of the 20 that minimized the following reconstruction error: 
\begin{equation*}
min  \left\{ \dfrac{\| f \otimes f - r \otimes r \|_{F}} {\| f \otimes f\|_{F}}, \dfrac{\| f \otimes f - \overline{r} \otimes \overline{r} \|_{F}} {\| f \otimes f\|_{F}} \right\}\end{equation*}
where $\| \cdot \|_{F}$ is the Fr\"obenius norm, and $f$ and $r$ are the vector of samples
\[ ( f(-20), f(-19.5), f(-19), \dots, f(19.5), f(20) )^{T}, \quad ( r(-20), r(-19.5), r(-19), \dots, r(19.5), r(20) )^{T} \]
respectively.  Over the 100 instances we found that the largest relative reconstruction error (after choosing the minimizer over the 20 applications of Algorithm \ref{RA:conv}) was $0.0504$.

The main source of error in our experiments seem to be Step 6 in Algorithm \ref{RA:conv}.  In all of our instances, the function $f$ has the property that $f(n) = 0$ for $|n| > 10$, so we might say that it is sparse in the standard basis $\{ \text{sinc }(t - n) \}_{n \in \mathbb{Z}}$.  However, to avoid the fact that $f(0) = 0$, we shift the reconstruction samples to $\{ \frac{n}{2} - \beta \}_{n \in \mathbb{Z}}$, but $f$ is not sparse in this coordinate system (frame) on $PW_{\pi}$.  Since we only utilize the samples $\{ \frac{n}{2} - \beta \}_{|n| \leq 20}$, we lose some of the energy of $f$ from this truncation.


We illustrate one instance of the signal reconstruction in Figure \ref{fig:instance}.  MATLAB code for these numerical experiments are available at \texttt{bitbucket.org/esweber/conjugate-phase-retrieval/}.

\begin{figure}[h]
\begin{tabular}{cc}
\includegraphics[width=8cm]{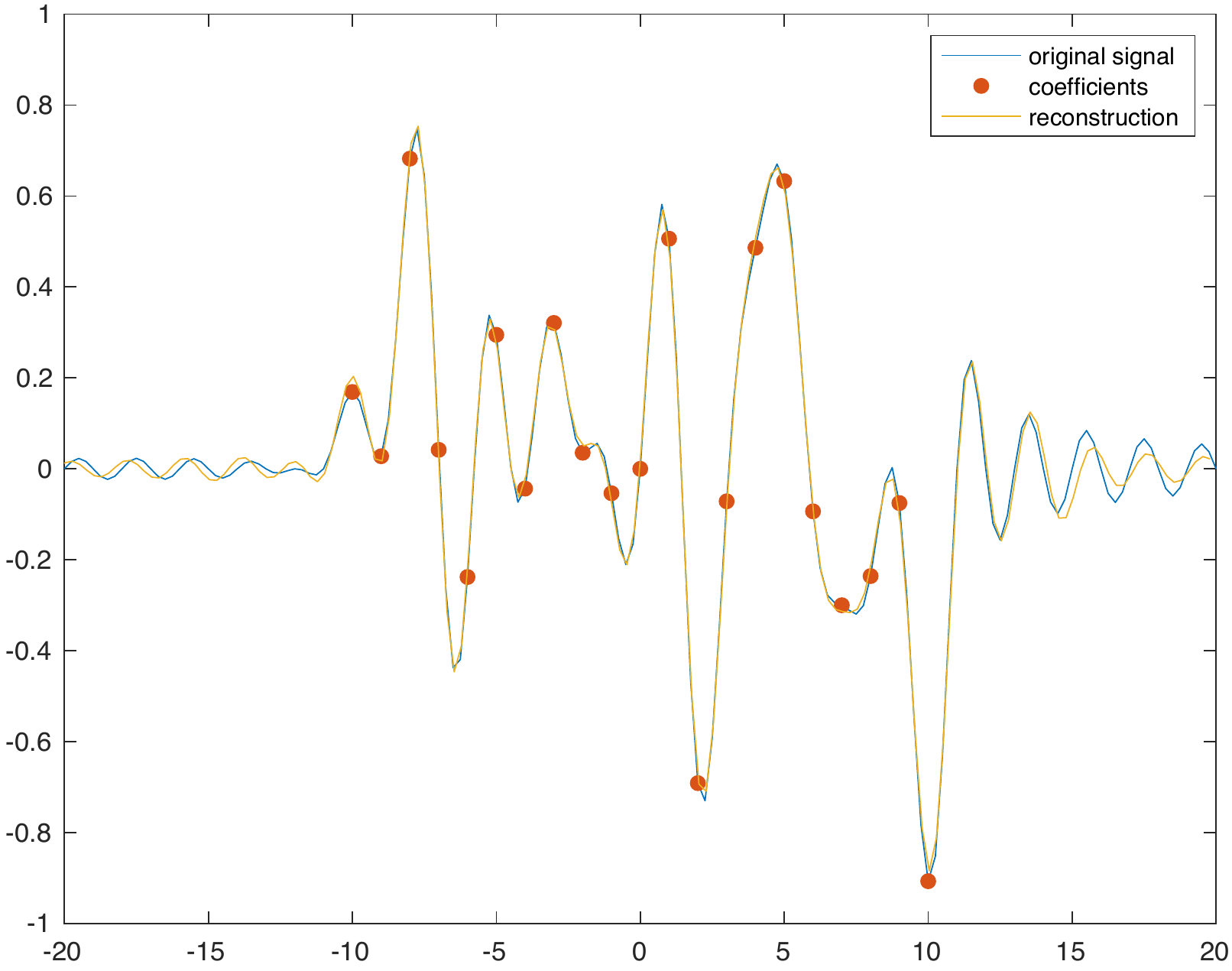} & \includegraphics[width=8cm]{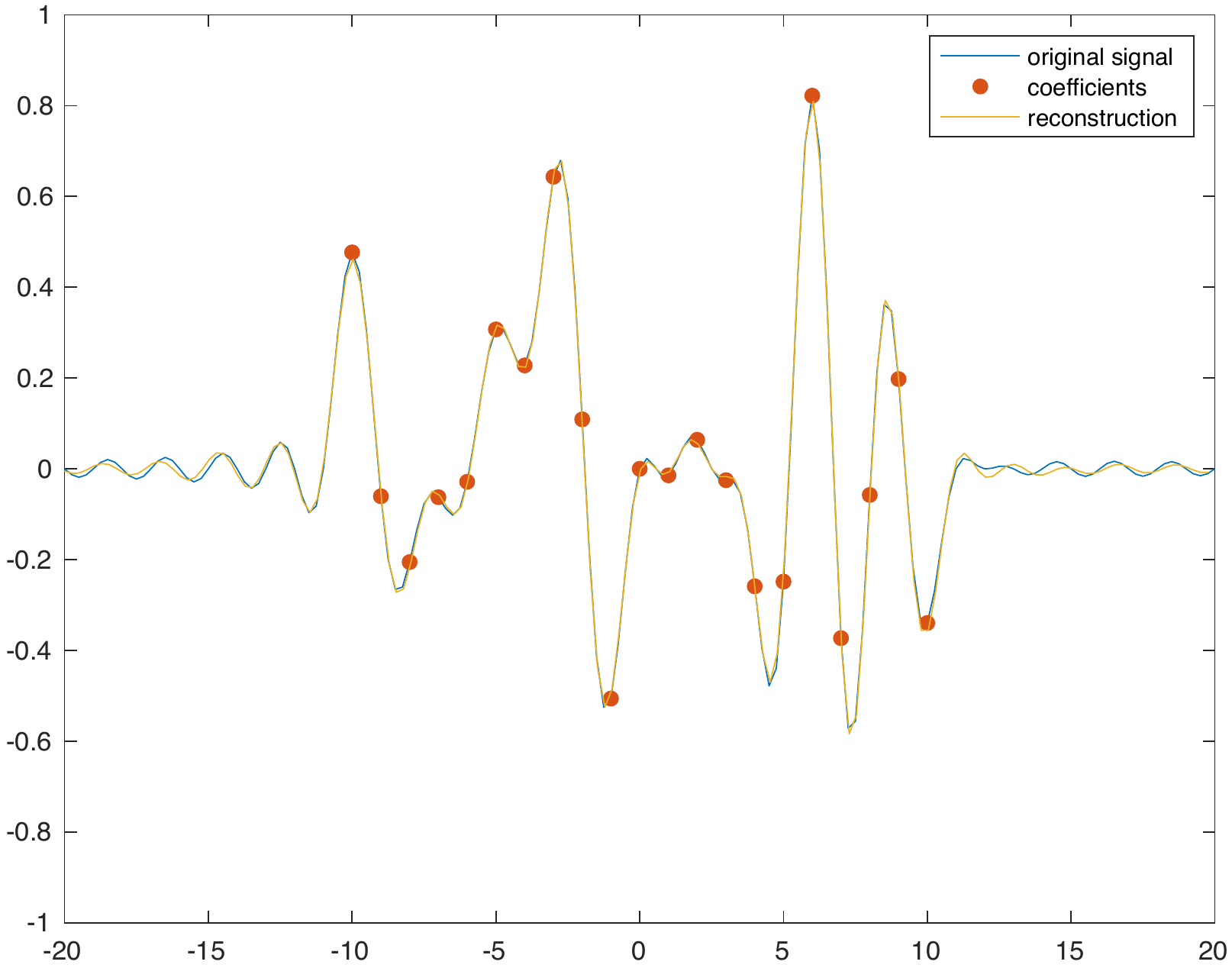} \\
Real part & Imaginary part
\end{tabular}
\caption{Reconstruction error: ($f$-original signal; $r$-reconstructed signal)
\[ \| f \otimes f - r \otimes r \|_{F}/ \| f \otimes f \|_{F} = 0.026 \]
$\beta = 0.2119$ (chosen using the MATLAB command \texttt{rand}).}
\label{fig:instance}
\end{figure}

\section{Appendix}

\subsection*{Continuity of the Reconstruction} \label{ssec:continuity}

Given the equivalence relation defined by Equation (\ref{Eq:equiv}), we have by Theorem \ref{Th:main} that for sequences $\{ t_{n} \}$ that are sampling sequences for $\mathcal{H}(E^2)$, the mapping
\begin{equation} \label{Eq:transform}
\widetilde{\mathcal{A}} : \mathcal{H}(E)/ \sim \ \to \ell^{2}(\mathbb{Z}) \oplus \ell^{2}(\mathbb{Z}) : f \mapsto ( | f(t_{n}) |, | f'(t_{n}) | )_{n}
\end{equation}
is one-to-one.  

We endow the quotient space $\mathcal{H}(E)/ \sim$ with the natural metric
\begin{equation} \label{Eq:metric}
d([f],[g]) := \inf \{ \| h_{1} - h_{2} \| : h_{1} \in [f], \ h_{2} \in [g] \} = \inf \{ \{ \| f - \alpha g \|, \| f - \alpha g^{\sharp} \| : | \alpha | = 1 \}. 
\end{equation}
In this metric, $\widetilde{\mathcal{A}}$ is continuous, since it is the composition of two continuous maps, namely the frame operator and the absolute value.  Following the idea of \cite{MW15a}, we demonstrate that the inverse is also continuous.

\begin{proposition}
The range $\mathcal{R}(\widetilde{\mathcal{A}})$ is closed.  The inverse of $\widetilde{\mathcal{A}}$ is continuous from $\mathcal{R}(\widetilde{\mathcal{A}})$ to $\mathcal{H}(E)/ \sim$.
\end{proposition}

\begin{proof}
Assume the sequence $\{ ( a_{n}^{(k)}, b_{n}^{(k)} )_{n} \}_{k} \subset \mathcal{R}(\widetilde{\mathcal{A}})$ converges in $\ell^{2}(\mathbb{Z}) \oplus \ell^{2}(\mathbb{Z})$ to $( a_{n}^{(0)}, b_{n}^{(0)})_{n}$.  For each $k$, there exists a $f_{k} \in \mathcal{H}(E)$ such that $ ( a_{n}^{(k)}, b_{n}^{(k)} )_{n} = ( | f_{k}(t_{n}) |, | f'_{k}(t_{n})|)_{n}$; for convenience, denote $ ( f_{k}(t_{n}) ,  f'_{k}(t_{n}) )_{n} = (\alpha_{n}^{(k)}, \beta_{n}^{(k)})_{n}$. For each fixed $n$, the sequence $\{ (\alpha_{n}^{(k)}, \beta_{n}^{(k)} ) \}_{k}$ has a convergent subsequence in $\mathbb{C}^{2}$; by a standard diagonalization argument there exists a subsequence that converges for every $n$.  Denote this subsequence by $k_{j}$, and the limit $( \alpha_{n}^{(0)}, \beta_{n}^{(0)})$.  We claim that
\[ \lim_{j} ( \alpha_{n}^{k_{j}}, \beta_{n}^{k_{j}} )_{n} = (\alpha_{n}^{(0)}, \beta_{n}^{(0)})_{n} \]
in the $\ell^2$-norm.

For $N \in \mathbb{N}$,
\begin{align*}
\sqrt{ \sum_{|n| \geq N} | \alpha_{n}^{(k_{j})}|^2 + | \beta_{n}^{(k_{j})} |^2 } &= \sqrt{ \sum_{|n| \geq N} | a_{n}^{(k_{j})}|^2 + | b_{n}^{(k_{j})} |^2 } \\
&\leq \sqrt{\sum_{|n| \geq N} | a_{n}^{(k_{j})} - a_{n}^{0} |^2 + | b_{n}^{(k_{j})} -  b_{n}^{0} |^2 } + \sqrt{\sum_{|n| \geq N} | a_{n}^{0} |^2 + | b_{n}^{0} |^2 }
\end{align*}
We are assuming that $\sum_{n \in \mathbb{Z}} | a_{n}^{(k_{j})} - a_{n}^{0} |^2 + | b_{n}^{(k_{j})} -  b_{n}^{0} |^2 \to 0$ as $j \to \infty$, so we have that
\[ \limsup_{j \to \infty} \sqrt{ \sum_{|n| \geq N} | \alpha_{n}^{(k_{j})}|^2 + | \beta_{n}^{(k_{j})} |^2 } \leq \sqrt{\sum_{|n| \geq N} | a_{n}^{0} |^2 + | b_{n}^{0} |^2 }. \]
It follows that
\begin{align*} 
\limsup_{j \to \infty} & \sum_{n \in \mathbb{Z}} | \alpha_{n}^{(k_{j})} - \alpha_{n}^{(0)} |^2 + | \beta_{n}^{(k_{j})} - \beta_{n}^{(0)} |^2  \\ 
&\leq  \limsup_{j \to \infty}  \sum_{|n| < N} | \alpha_{n}^{(k_{j})} - \alpha_{n}^{(0)} |^2 + | \beta_{n}^{(k_{j})} - \beta_{n}^{(0)} |^2  + 
		\limsup_{j \to \infty} \sum_{|n| \geq N} | \alpha_{n}^{(k_{j})} - \alpha_{n}^{(0)} |^2 + | \beta_{n}^{(k_{j})} - \beta_{n}^{(0)} |^2 \\
&= \limsup_{j \to \infty} \sum_{|n| \geq N} | \alpha_{n}^{(k_{j})} - \alpha_{n}^{(0)} |^2 + | \beta_{n}^{(k_{j})} - \beta_{n}^{(0)} |^2 \\
&\leq \left( \sum_{|n| \geq N} | a_{n}^{(0)} |^2 + | b_{n}^{(0)} |^2 + \sum_{|n| \geq N} | \alpha_{n}^{(0)} |^2 + | \beta_{n}^{(0)} |^2 \right).
\end{align*}
Therefore,
\begin{align*}
\limsup_{j \to \infty} & \sum_{n \in \mathbb{Z}} | \alpha_{n}^{(k_{j})} - \alpha_{n}^{(0)} |^2 + | \beta_{n}^{(k_{j})} - \beta_{n}^{(0)} |^2 \\
&\leq \limsup_{N \to \infty} \left( \sum_{|n| \geq N} | a_{n}^{(0)} |^2 + | b_{n}^{(0)} |^2 + \sum_{|n| \geq N} | \alpha_{n}^{(0)} |^2 + | \beta_{n}^{(0)} |^2 \right) \\
&= 0.
\end{align*}
The completes the claim.  We have that the sequence $ \{ ( \alpha_{n}^{(k_{j})}, \beta_{n}^{(k_{j})})_{n} \}_{j} $ is contained in the image of the sampling transform, which has closed range, and therefore $(\alpha_{n}^{(0)}, \beta_{n}^{(0)})_{n}$ is also in the range of the sampling transform, whence there exists an $f_{0} \in \mathcal{H}(E)$ such that $((\alpha_{n}^{(0)}, \beta_{n}^{(0)})_{n} = ( f_{0} (t_{n}), f'_{0}(t_{n}) )_{n}$ from which we obtain that $(a_{n}^{(0)}, b_{n}^{(0)})_{n} = ( | f_{0} (t_{n}) |, |f'_{0}(t_{n}) | )_{n} \in \mathcal{R}(\widetilde{\mathcal{A}})$.  This concludes the proof of the first part.

Now for the continuity:  (outline)
\begin{enumerate}
\item  Fix a sequence $\vec{v}_{n}$ of elements in $\mathcal{R}(\widetilde{\mathcal{A})}$ that converges.
\item  For each such element, pick a representative $f_{n}$ where $\widetilde{\mathcal{A}} ([f_{n}]) = \vec{v}_{n}$.
\item  For each subsequence of $\{ f_{n} \}$, there exists a subsequence such that $ \{ \Phi (f_{n_{j_{k}}}) \}$ converges in $\ell^2$.
\item  For this subsequence, $\{ f_{n_{j_{k}}} \}$ converges in $\mathcal{H}(E)$, therefore $ [ f_{n_{j_{k}}} ] \to [ f ] $.
\end{enumerate}

To prove continuity, assume $\widetilde{\mathcal{A}}( [ f_{k} ] ) = ( a_{n}^{(k)}, b_{n}^{(k)} )_{n}$ converges to $ ( a_{n}^{(0)}, b_{n}^{(0)} )_{n} = \widetilde{\mathcal{A}}([f_{0}])$ in the $\ell^{2}$-norm.  We prove that every subsequence $ [ f_{k_{j}} ]$ has a subsequence that converges to $ [ f_{0} ]$.  As before, the sequence
\[ \Phi( f_{k_{j}} ) = ( f_{k_{j}} (t_{n}), f'_{k_{j}} (t_{n}) )_{n} \]
has a subsequence $f_{k_{j_{l}}}$ such that $\Phi( f_{k_{j_{l}}} )$ converges pointwise to a sequence $( \alpha_{n}^{(0)}, \beta_{n}^{(0)} )$, which is $\Phi(f)$ for some $f$.  Note that $( | f(t_{n}) |, | f'(t_{n}) | )_{n} = (a_{n}^{(0)}, b_{n}^{(0)})_{n}$, so $f \in [ f_{0} ]$.  Now, again by above, we have that $\Phi(f_{k_{j_{l}}})$ converges to $\Phi(f)$ in the $\ell^2$-norm; since $\Phi$ has a continuous inverse, we have $f_{k_{j_{l}}}$ converges to $f$ in $\mathcal{H}(E)$.  It follows that $[ f_{k_{j_{l}}} ] \to [f_{0}]$, completing the proof.

\end{proof}

\subsection*{Conjugate Phase Retrieval in Other Spaces}

There are other natural spaces for which it may be possible to extend our methods  (see also \cite{chen2019phase} for related results).  In particular, spaces whose elements are entire functions are natural to consider, since our methods utilized properties of zeros of entire functions.  Note that other properties of $PW_{\pi}$ we used include: there are sets of sampling for $PW_{\pi}$ that have regular structure (in particular, finite unions of lattices); the squares of elements in $PW_{\pi}$ lie in a space that also have sets of sampling; and $PW_{\pi}$ is closed under translations.  Spaces that are natural to consider include $PW_{\gamma}^{p}$; Bernstein spaces \cite{PYB14a}; de Branges spaces \cite{goncalves2018interpolation}; and generalized Paley-Wiener spaces as defined in \cite{weber2019paley}.  None of these spaces satisfy all of the properties of $PW_{\gamma}$ that we use in this paper.  The generalized Paley-Wiener spaces need not be closed under $f \mapsto f^{\sharp}$, but do admit a sampling theory \cite{HW17a}.

\subsection*{Concluding Remarks}
Code for numerical experiments in Subsections \ref{ssec:ap} and \ref{ssec:numerics} is available from 
\texttt{bitbucket.org/esweber/conjugate-phase-retrieval/}.

Acknowledgements:  Eric Weber was supported in part by the National Science Foundation under award \#1934884 and the National Geospatial-Intelligence Agency under award \#1830254.

\bibliographystyle{amsplain}
\bibliography{cprpws}
\nocite{*}

\end{document}